\documentclass [6pt,a4paper]{article}
\usepackage{bbm}
\usepackage{amsfonts}
\usepackage{mathrsfs}
\usepackage {amssymb}
\usepackage {amsmath}
\usepackage {ntheorem}
\usepackage{latexsym}
\usepackage{booktabs}
\usepackage{mathrsfs}
\usepackage{bm}
\newenvironment{proof}{{\noindent\it\textbf{Proof}}\quad}{\hfill $\square$\par}
\usepackage{lastpage}
\usepackage{epsfig}

\begin{document}\large
\newtheorem{lemma}{Lemma}[section]
\newtheorem{theorem}[lemma]{Theorem}
\newtheorem{example}[lemma]{Example}
\newtheorem{definition}[lemma]{Definition}
\newtheorem{proposition}[lemma]{Proposition}
\newtheorem{conjecture}[lemma]{Conjecture}
\newtheorem{corollary}[lemma]{Corollary}
\newtheorem{remark}{Remark}
%********************************************************************************************************
\begin{center}
\textbf{\LARGE{New entanglement-assisted quantum MDS codes with length $n=\frac{q^2+1}5$}}\footnote { E-mail
addresses: zhushixinmath@hfut.edu.cn(S.Zhu),jiangw000@163.com(W.Jiang),chenxiaojing0909@ahu.\\edu.cn (X.Chen).This research is supported by the National Natural Science
 Foundation of China ( No.61772168).
}\\
\end{center}

\begin{center}
{ { Shixin Zhu$^1$, \  Wan Jiang$^1$, \  Xiaojing Chen$^2$} }
\end{center}

\begin{center}
\textit{1\ School of Mathematics, Hefei University of
Technology, Hefei 230009, Anhui, P.R.China\\
2\  School of Internet, Anhui University, Hefei 230039, Anhui, P.R.China}
\end{center}

\noindent\textbf{Abstract:} The entanglement-assisted stabilizer formalism can transform arbitrary classical linear codes into entanglement-assisted quantum error correcting codes (EAQECCs). In this work, we  construct some new entanglement-assisted quantum MDS (EAQMDS) codes with length  $n=\frac{q^2+1}5$ from cyclic codes. Compared with all the previously known parameters with the same length, all of them have flexible parameters and larger minimum distance. 

\section{Introduction}

In recent years, quantum computation has become a hot research issue. As in the case of classical digital communications, quantum error-correcting codes (QECCs) play an important role in quantum information processing and quantum computation \cite{ref1}-\cite{ref2}. There are many good works about QECCs \cite{ref3}-\cite{ref27}.

In kinds of construction methods of QECCs, the CSS construction is most frequently used. It establishes a relationship between classical error correcting codes and QECCs, but it needs classical codes to be dual-containing or self-orthogonal which is not easy to satisfy all the time. This problem is solved after Brun et al. proposed EAQECCs which allows non-dual-containing classical codes to construct QECCs if the sender and receiver shared entanglement bits in advance \cite{ref4}. This has inspired more and more scholars to focus on constructing EAQECCs \cite{ref16}-\cite{ref17}.

As we all know, many EAQMDS codes with a small number of entangled states have been constructed. In \cite{ref21}, Lu et al. constructed new EAQMDS codes of length $n = q^2+1$ with larger minimum distance and consumed four entanglement bits. Besides, Chen et al. constructed EAQMDS codes by using constacyclic codes with length $n=\frac{q^2+1}{5}$ and consumed four entanglement bits in \cite{ref28}. Actually, the larger the minimum distance of EAQECCs are, the more the entanglement bits will be employed. However, it is not an easy task to analyze the accurate parameters if the value of $c$ is too large or flexible. Recently, some scholars have obtained great progress. In \cite{ref23}, Qian and Zhang constructed some new EAQMDS codes with length $n = q^2+1$ and some new entanglement-assisted quantum almost MDS codes with length $n = q^4-1$. Besides, Wang et al. obtained series of EAQECCs with flexible parameters of length $n = q^2+1$ by using $q^2$-cyclotomic coset modulo $rn$ \cite{ref24}. And almost all of those known results about EAQECCs with the same length are some special cases of their. 

Inspired by the above work, we consider to use cyclic codes to construct EAQECCs of length $n=\frac{q^2+1}{5}$, with flexible parameters naturally. Based on cyclic codes, we construct four classes of EAQMDS codes with the following parameters:

(1) $[[n, n-4(m-1)(5m-q-5)-1, 2(m-1)q+2; 20(m-1)^2+1]]_q,$ where $ q=10k+3 $ $(k\geq 2) $  is an odd prime power and  $2\leq m\leq \frac{q-3}{10}$.

(2) $[[n, n-4(m-1)(5m-q-5)-1, 2(m-1)q+2; 20(m-1)^2+1]]_q,$ where $ q=10k+7 $ $(k\geq 2) $  is an odd prime power and  $2\leq m\leq \frac{q-7}{10}$.

(3) $[[n, n-4(m-1)(5m-q-5)-1, 2(m-1)q+2; 20(m-1)^2+1]]_q$, where $ q=2^e(e>1)  $,  $e\equiv 1~ {\rm mod~ }4 $ is a positive integer and $2\leq m\leq \frac{q-2}{10}$.

(4) $[[n, n-4(m-1)(5m-q-5)-1, 2(m-1)q+2; 20(m-1)^2+1]]_q$, where $ q=2^e(e>1)  $, $e\equiv 3~ {\rm mod~ }4 $ is a positive integer and $2\leq m\leq \frac{q-8}{10}$.

The main organization of this paper is as follows. In Sect.2, some basic background and results about cyclic codes and EAQECCs are reviewed. In Sect.3, we construct four classes of optimal EAQMDS codes with length $n=\frac{q^2+1}{5}$. Sect.4 concludes the paper.

\section{Preliminaries} 
  In this section, we will review some relevant concepts on cyclic codes and EAQECCs. For further and detailed information on cyclic codes can be found in \cite{ref25,ref29}, EAQECCs please see \cite{ref5, ref4, ref11, ref23, ref24, ref31}.
\subsection{Review of Cyclic Codes}
  For given a positive integer $l$ and prime number $p$, let $q=p^l$ and $\mathbb{F}_{q^2}$ be the finite field of $q^2$ elements. A $k$-dimensional subspace of the $n$-dimensional vector space $\mathbb{F}_{q^2}^n$ is a linear code of length $n$ over $\mathbb{F}_{q^2}$ and this linear code is denoted by $[n,k]_{q^2}$. If an $[n,k]_{q^2}$ linear code $\mathcal{C}$ can detect $d-1$ errors but not $d$ errors, it is an $[n,k,d]_{q^2}$ linear code $\mathcal{C}$. For any $\alpha \in \mathbb{F}_{q^2}$, the conjugation of $\alpha$ is denoted by $\overline\alpha=\alpha^q$ and the conjugation transpose of an $m\times n$ matrix $H=(x_{i,j})$ entries in $\mathbb{F}_{q^2}$ is an $n \times m$ matrix $H^\dag=({x_{j,i}^q})$.
  
  The Hermitian inner product $\langle\mu,\upsilon\rangle_h$ of the vectors $\mu=(u_0,u_1,\cdots,u_{n-1})$ and $\upsilon=(v_0,v_1,\cdots,v_{n-1})$ in $\mathbb{F}_{q^2}^n$ is $$\langle\mu,\upsilon\rangle_h=\sum_{i=0}^{n-1}\overline{u_i}v_i=u_0^qv_0+u_1^qv_1+\cdots+u_{n-1}^qv_{n-1}.$$
  The Hermitian dual code of $\mathcal{C}$ is defined as
  $$\mathcal{C}^{\perp_h}=\{\mu\in\mathbb{F}_{q^2}^n~|\langle\mu,\upsilon\rangle_h=0~{\rm for~ all}~\upsilon\in \mathcal{C}\}.$$
 
  If $\mathcal{C}\subseteq \mathcal{C}^{\perp_h}$, then $\mathcal{C}$ is called a Hermitian self-orthogonal code.
  If $\mathcal{C}^{\perp_h}\subseteq\mathcal{C}$, then $\mathcal{C}$ is called a Hermitian dual-containing code.
  
  A linear code $\mathcal{C}$ of length $n$ over $\mathbb{F}_{q^2}^n$ is said to be cyclic if for any codeword  $(c_0,c_1,\ldots,c_{n-1})\in \mathcal{C}$ implies its cyclic shift
  $(c_{n-1},c_0,\dots,c_{n-2})\in \mathcal{C}$. For a cyclic code $\mathcal{C}$, each codeword $c=(c_0,c_1,\cdots,c_{n-1})$ is customarily represented in its polynomial form: $c(x)=c_0+c_1x+\dots+c_{n-1}x^{n-1}$, and the $\mathcal{C}$ is in turn identified with the set of all polynomial representations of its codewords. Then, a $q^2$-ary cyclic code $\mathcal{C}$ of length $n$ is an ideal of $\mathbb{F}_{q^2}[x]/\langle x^n-1 \rangle $ and $\mathcal{C}$ can be generated by a monic polynomial factors of $x^n-1$, i.e. $C=\langle g(x) \rangle$ and $g(x)|(x^n-1)$. The $g(x)$ is called the generator polynomial of $\mathcal{C}$. 
  
  Note that $x^n-1$ has no repeated root over $\mathbb{F}_{q^2}$ if and only if gcd$(n,q)=1$. Let $ \lambda $ denote a primitive $n$-th root of unity in some extension field of $\mathbb{F}_{q^2}$. Hence, $ x^n-1=\prod_{i=0}^{n-1}(x-\lambda^i)$. For $ 0\leq i \leq n-1$, the $q^2$-cyclotomic coset modulo $n$ containing $i$ is defined by the set $$C_i=\{i, iq^2,iq^4,...,iq^{2(m_i-1)} \},$$ where $m_i$ is the smallest positive integer such that $iq^{2m_i}\equiv i~mod~n$. Each $C_i$ corresponds to an irreducible divisor of $x^n-1$ over $\mathbb{F}_{q^2}$. The defining set of a cyclic code $\mathcal{C}=\langle g(x)\rangle$ of length $n$ is the set $Z=\{0\leq i\leq n-1~|~g(\lambda^i)=0\}$.
  
  Let $\mathcal{C}$ be an $[n,k,d]$ cyclic code over $\mathbb{F}_{q^2}$ with defining set $Z$. Obviously $Z$ must be a union of some $q^2$-cyclotomic coset modulo $n$ and $dim(\mathcal{C})=n-|Z|$.  It is clear to see that the Hermitian dual code $\mathcal{C}^{\perp_h}$ has a defining set $Z^{\perp_h}=\{0\leq z \leq n-1|-qz~mod~n \notin Z\}$. Note that $Z^{-q}=\{-qz~mod~n|~z\in Z\}$. Then, $\mathcal{C}$ contains its Hermitian dual
  code if and only if $Z\bigcap Z^{-q}=\emptyset$ from Lemma 2.2 in \cite{ref25}.
  The following is a lower bound for cyclic codes:
 
 \begin{proposition}\label{pro:2.1} \rm\cite{ref29}( The BCH Bound for Cyclic Codes)\emph{ Let $\mathcal{C}$ be a $q^2$-ary cyclic code of length $n$ with defining set $Z$. If $Z$ contains $d-1$ consecutive elements, then the minimum distance of $\mathcal{C}$ is at least $d$.}
\end{proposition}

\subsection{Review of EAQECCs }
 
 Let $c$ be a nonnegative integer. Via $c$ pairs of maximally entanglement states, an  $[[n,k, d; c]]_q$ EAQECC encodes $k$ information qubits into $n$ qubits, and it can correct up to $\lfloor \frac{d-1}{2}\rfloor$ errors which act on $n$ qubits, where $d$ is called minimum distance of the EAQECC.
 
  A parity check $H$ of an $[n,k]_{q^2}$ linear code $\mathcal{C}$ with respect to Hermitian inner product is an $(n-k)\times n$ matrix whose rows constitute a basis for $\mathcal{C}^{\perp_h}$, then $\mathcal{C}^{\perp_h}$ has an  $n \times (n-k) $ parity check matrix $H^\dag$.
  
  Brun et al. established a bound on the parameters of an  $[[n,k, d; c]]_q$ EAQECC.
 \begin{proposition}\label{pro:2.3} \rm\cite{ref4,ref11}  
 	\emph{Assume that $\mathcal{C}$ is an entanglement-assisted quantum code with parameters $[[n,k, d; c]]_q$. If $d\leq \frac{n+2}{2} $, then $\mathcal{C}$ satisfies the entanglement-assisted Singleton bound $n+c-k\geq 2(d-1)$. If $\mathcal{C}$ satisfies the equality $n+c-k= 2(d-1)$ for $d\leq \frac{n+2}{2} $, then it is called an entanglement-assisted quantum MDS code.}
\end{proposition}

 	According to literature \cite{ref5, ref4,ref31}, EAQECCs can be constructed from arbitrary linear codes over $\mathbb{F}_{q^2}$, which is given by the following proposition.
 	
\begin{proposition}\label{pro:2.4} 
 		If $\mathcal{C}$ is an $[n,k,d]_{q^2}$ classical code and $H$ is its parity check matrix over $\mathbb{F}_{q^2}$, then there exist entanglement-assisted quantum codes with  parameters $[[n,2k-n, d; c]]_q$, where $c = {\rm rank}(HH^\dag)$. 
\end{proposition}
 	
 	Although it is possible to construct an EAQECCs from any classical linear code over $\mathbb{F}_{q^2}$, it is not easy to calculate the parameter of ebits $c$. However,  $c$ can be easily determined for some special classes of linear codes. In \cite{ref5} defining the decomposition of the defining set of cyclic codes was initially introduced.
 	
\begin{definition}\label{def:2.5}  \rm\cite{ref5}
	\emph{Let $\mathcal{C}$ be a $q^2$-ary cyclic code of length $n$ with defining set $Z$. Assume
	that $Z_2 = Z \bigcap (-qZ)$ and $Z_1 = Z\backslash Z_2$, where $-qZ = \{n - qx|x \in Z\}$. Then,
	$Z = Z_1 \bigcup Z_2$ is called a decomposition of the defining set of $\mathcal{C}$.}
\end{definition} 

\begin{lemma}\label{le:2.6}
	Let $\mathcal{C}$ be a cyclic code with length n over  $\mathbb{F}_{q^2}$, where $\gcd(n, q) = 1$.
	Suppose that $Z$ is the defining set of the cyclic code $\mathcal{C}$ and $Z = Z_1 \bigcup Z_2$ is a
	decomposition of $Z$. Then, the number of entangled states required is $c = |Z_2|$.	
\end{lemma}

\section{New EAQMDS codes of length ${\frac{q^2+1}{5}}$}
In this section, we use cyclic codes of length $n=\frac{q^2+1}{5}$ to construct some new EAQMDS codes. Let $n=\frac{q^2+1}{5}$, where $ q\geq 5 $ is an odd prime power, it is clear that that the $q^2$-cyclotomic cosets modulo $n$  are 
$$C_0=\{0\}, C_1=\{1,-1\}=\{1,n-1\}, C_2=\{2,-2\}=\{2,n-2\},...~, C_{\frac{q^2+1}{10}}=\{\frac{q^2+1}{10}\}.$$

Now, we first give a useful lemma that will be used in later constructions.

\begin{lemma}\label{le:3.2}
	Let $n=\frac{q^2+1}{5}$ and  $ q $  be an odd prime power.\\
	1) When $ q=10k+3~(k\geq 2) $, then $-qC_{sq+i}={C_{iq-s}}$, where $1\leq i\leq\frac{3q-9}{10} $ and \  $\frac{2q+4}{5}\leq i\leq \frac{3q-4}{5} $, if $0\leq s\leq\frac{q-3}{10}-1$, or $1\leq i\leq\frac{3q-9}{10} $, if $ s=\frac{q-3}{10}$.\\\\
	2) When $ q=10k+7~(k\geq 2) $, then $-qC_{sq+i}={C_{iq-s}}$, where $1\leq i\leq\frac{q-2}{5} $, $\frac{3q+9}{10}\leq i\leq\frac{2q-4}{5} $ and $\frac{q+1}{2}\leq i\leq\frac{7q-9}{10},$ if $0\leq s\leq\frac{q-7}{10}$.
	
\end{lemma} 

\begin{proof} 
	1) Note that  $C_{sq+i}=\{sq+i,-(sq+i)\}$  for $1\leq i\leq\frac{3q-9}{10} $ and $\frac{2q+4}{5}\leq i\leq\frac{3q-4}{5} ,$ if $  0\leq s\leq\frac{q-3}{10}-1$, or  $1\leq i\leq\frac{3q-9}{10} ,$ if $ s=\frac{q-3}{10}$.
	
	Since $-q\cdot(-(sq+i))=sq^2+iq=s(q^2+1)+iq-s\equiv iq-s~ mod~n$. This gives that $-qC_{sq+i}={C_{iq-s}}$.
	
	2) The proof is similar to case 1), so it is omitted here.

\end{proof}
    \ 
    
    From Lemma \ref{le:3.2}, we can also obtain $-qC_{tq-j}={C_{jq+t}}$, where the range of $q$, $j$ and $t$ is listed below:
    
    1) When $ q=10k+3~(k\geq 2) $, we have $1\leq t\leq\frac{3q-9}{10} $ and $\frac{2q+4}{5}+2\leq t \leq\frac{3q-4}{5} ,$ if $  0\leq j\leq\frac{q-3}{10}-1$, or  $1\leq t\leq\frac{3q-9}{10} ,$ if $ j=\frac{q-3}{10}$.
    
    2) When $ q=10k+7~(k\geq 2) $, we have $1\leq t\leq\frac{q-2}{5} $, $\frac{3q+9}{10}\leq t\leq\frac{2q-4}{5} $ and $\frac{q+1}{2}\leq t\leq\frac{7q-9}{10},$ if \  $0\leq j\leq\frac{q-7}{10}$.\\
    
    Based on the discussions above, we can give the first construction as follows.\\\\
\noindent\textbf{Case \uppercase\expandafter{\romannumeral 1} \ \ $q=10k+3$}\\
    
    In order to determine the number of entangled states $c$, we give the following Lemma for preparation.

\begin{lemma}\label{le:3.3}
	Let $n=\frac{q^2+1}{5}$ and  $ q=10k+3 $ $(k\geq 2) $  be an odd prime power. For a positive integer $2\leq m\leq \frac{q-3}{10}$, let 
	\\$$Z_{1}=
	\bigcup_{\substack{m\leq i_1 \leq \frac{q+2}{5}-m,\\0\leq s\leq m-2}}C_{sq+i_1}
	\bigcup_{\substack{\frac{q-3}{5}+m\leq i_2 \leq \frac{2q+4}{5}-m,\\0\leq s\leq m-2}}C_{sq+i_2}
	\bigcup_{\substack{\frac{2q-1}{5}+m\leq i_3 \leq \frac{3q+1}{5}-m,\\0\leq s\leq m-2}}C_{sq+i_3}$$
	$$\bigcup_{\substack{\frac{q-3}{5}+m\leq j_1 \leq \frac{2q+4}{5}-m,\\1\leq t\leq m-1}}C_{tq-j_1}
	\bigcup_{\substack{m-1\leq j_2 \leq \frac{q+2}{5}-m,\\1\leq t\leq m-1}}C_{tq-j_2}.
	$$
	Then $Z_1\bigcap-qZ_1=\emptyset$.
	
\end{lemma}
\begin{proof}
	For a positive integer $m$ with $2 \leq m\leq \frac{q-3}{10} $, let 
	\\$$Z_{1}=
	\bigcup_{\substack{m\leq i_1 \leq \frac{q+2}{5}-m,\\0\leq s\leq m-2}}C_{sq+i_1}
	\bigcup_{\substack{\frac{q-3}{5}+m\leq i_2 \leq \frac{2q+4}{5}-m,\\0\leq s\leq m-2}}C_{sq+i_2}
	\bigcup_{\substack{\frac{2q-1}{5}+m\leq i_3 \leq \frac{3q+1}{5}-m,\\0\leq s\leq m-2}}C_{sq+i_3}$$
	$$\bigcup_{\substack{\frac{q-3}{5}+m\leq j_1 \leq \frac{2q+4}{5}-m,\\1\leq t\leq m-1}}C_{tq-j_1}
	\bigcup_{\substack{m-1\leq j_2 \leq \frac{q+2}{5}-m,\\1\leq t\leq m-1}}C_{tq-j_2}.
	$$
	Then by Lemma \ref{le:3.2}, we have
	$$-qZ_{1}=
	\bigcup_{\substack{m\leq i_1 \leq \frac{q+2}{5}-m,\\0\leq s\leq m-2}}C_{i_1q-s}
	\bigcup_{\substack{\frac{q-3}{5}+m\leq i_2 \leq \frac{2q+4}{5}-m,\\0\leq s\leq m-2}}C_{i_2q-s}
	\bigcup_{\substack{\frac{2q-1}{5}+m\leq i_3 \leq \frac{3q+1}{5}-m,\\0\leq s\leq m-2}}C_{i_3q-s}$$
	$$\bigcup_{\substack{\frac{q-3}{5}+m\leq j_1 \leq \frac{2q+4}{5}-m,\\1\leq t\leq m-1}}C_{j_1q+t}
	\bigcup_{\substack{m-1\leq j_2 \leq \frac{q+2}{5}-m,\\1\leq t\leq m-1}}C_{j_2q+t}.
	$$
	When $m\leq i_1 \leq \frac{q+2}{5}-m,~0\leq s\leq m-2,$ it follows that
	$$sq+i_1\leq  (m-2)q+ \frac{q+2}{5}-m,~mq-m+2\leq i_1q-s.$$
	When $\frac{q-3}{5}+m\leq i_2 \leq \frac{2q+4}{5}-m,~0\leq s\leq m-2,$ it follows that
	$$sq+i_2\leq (m-2)q+\frac{2q+4}{5}-m,~(\frac{q-3}{5}+m)q-m+2\leq i_2q-s.$$
	When $\frac{2q-1}{5}+m\leq i_3 \leq \frac{3q+1}{5}-m,~0\leq s\leq m-2,$ it follows that
	$$sq+i_3\leq (m-2)q+\frac{3q+1}{5}-m,~(\frac{2q-1}{5}+m)q-m+2\leq i_3q-s.$$
	When $\frac{q-3}{5}+m\leq j_1 \leq \frac{2q+4}{5}-m,~1\leq t\leq m-1,$ it follows that
	$$tq-j_1\leq (m-1)q-\frac{q-3}{5}-m,~(\frac{q-3}{5}+m)q+1\leq j_1q+t.$$
	When $m-1\leq j_2 \leq \frac{q+2}{5}-m,~1\leq t\leq m-1,$ it follows that
	$$tq-j_2\leq (m-1)q-m+1,~(m-1)q+1\leq j_2q+t.$$
	
	It is easy to check that 
	$$sq+i_1<i_1q-s,sq+i_1<i_2q-s,sq+i_1<i_3q-s,sq+i_1<j_1q+t,sq+i_1<j_2q+t.$$
	$$sq+i_2<i_1q-s,sq+i_2<i_2q-s,sq+i_2<i_3q-s,sq+i_2<j_1q+t,sq+i_2<j_2q+t.$$
	$$sq+i_3<i_1q-s,sq+i_3<i_2q-s,sq+i_3<i_3q-s,sq+i_3<j_1q+t,sq+i_3<j_2q+t.$$
	$$tq-j_1<i_1q-s,tq-j_1<i_2q-s,tq-j_1<i_3q-s,tq-j_1<j_1q+t,tq-j_1<j_2q+t.$$
	$$tq-j_2<i_1q-s,tq-j_2<i_2q-s,tq-j_2<i_3q-s,tq-j_2<j_1q+t,tq-j_2<j_2q+t.$$
	
	For the range of $~i_1,~i_2,~i_3$ and $s$, note that
	$sq+i_r(r=1,2,3)\leq\frac{q^2-9}{10} $, the subscripts of $C_{sq+i_r}$ is the smallest number in the set. 
	Then $Z_1\bigcap-qZ_1=\emptyset$. The desired results follows.
	
\end{proof}
\ 

\noindent\textbf{Example 3.1}
    ~Let $q=23$ and $2\leq m\leq \frac{q-3}{10}=2.$
    Then, $n=\frac{q^2+1}{5}=106,~m=2.$
    According to Lemma \ref{le:3.3}, we can obtain
    \begin{equation*}
    \begin{split}
      Z_{1}&=
    \bigcup_{\substack{2\leq i_1 \leq 3,\\s=0}}C_{sq+i_1}
    \bigcup_{\substack{6\leq i_2 \leq 8,\\s=0}}C_{sq+i_2}
    \bigcup_{\substack{11\leq i_3 \leq 12,\\s=0}}C_{sq+i_3}
    \bigcup_{\substack{6\leq j_1 \leq 8,\\t=1}}C_{tq-j_1}
    \bigcup_{\substack{1\leq j_2 \leq 3,\\t=1}}C_{tq-j_2}\\
    &=C_2\bigcup C_3\bigcup C_6\bigcup C_7\bigcup C_8\bigcup C_{11}\bigcup C_{12}\bigcup C_{15}\bigcup C_{16}\bigcup C_{17}\bigcup C_{20}\bigcup C_{21}\\
    &~~~\bigcup C_{22}.
    \end{split}
    \end{equation*}

    It is easy to check that $Z_1\bigcap-qZ_1=\emptyset$.\\
    
\noindent\textbf{Example 3.2}
    ~Let $q=43$ and $2\leq m\leq \frac{q-3}{10}=4.$
    Then, $n=\frac{q^2+1}{5}=370,~m=2,~3,~4$  respectively.\\
    1) Let $m=2,$ according to Lemma \ref{le:3.3}, we can obtain 
\begin{equation*}
    \begin{split}
    Z_{1}&=
    \bigcup_{\substack{2\leq i_1 \leq 7,\\s=0}}C_{sq+i_1}
    \bigcup_{\substack{10\leq i_2 \leq 16,\\s=0}}C_{sq+i_2}
    \bigcup_{\substack{19\leq i_3 \leq 24,\\s=0}}C_{sq+i_3}
    \bigcup_{\substack{10\leq j_1 \leq 16,\\t=1}}C_{tq-j_1}
    \bigcup_{\substack{1\leq j_2 \leq 7,\\t=1}}C_{tq-j_2}\\
    &=C_2\bigcup C_3\bigcup\dots\bigcup C_7\bigcup C_{10}\bigcup\dots\bigcup C_{16}\bigcup C_{19}\bigcup\dots\bigcup C_{24}\bigcup C_{27}\bigcup\dots\\
    &~~~\bigcup C_{33}\bigcup C_{36}\bigcup\dots\bigcup C_{42}.    
    \end{split}
    \end{equation*}
 
    It is easy to check that $Z_1\bigcap-qZ_1=\emptyset.$ \\
    2) Let $m=3,$ according to Lemma \ref{le:3.3}, we can obtain
    \\$$Z_{1}=
    \bigcup_{\substack{3\leq i_1 \leq 6,\\0\leq s\leq 1}}C_{sq+i_1}
    \bigcup_{\substack{11\leq i_2 \leq 15,\\0\leq s\leq 1}}C_{sq+i_2}
    \bigcup_{\substack{20\leq i_3 \leq 23,\\0\leq s\leq 1}}C_{sq+i_3}
    \bigcup_{\substack{2\leq j_1 \leq 6,\\1\leq t\leq 2}}C_{tq-j_1}
    \bigcup_{\substack{11\leq j_2 \leq 15,\\1\leq t\leq 2}}C_{tq-j_2}.$$
    
    It is easy to check that $Z_1\bigcap-qZ_1=\emptyset$.\\
    3) Let $m=4$, according to Lemma \ref{le:3.3}, we can obtain
    \\$$Z_{1}=
    \bigcup_{\substack{4\leq i_1 \leq 5,\\0\leq s\leq 2}}C_{sq+i_1}
    \bigcup_{\substack{12\leq i_2 \leq 14,\\0\leq s\leq 2}}C_{sq+i_2}
    \bigcup_{\substack{21\leq i_3 \leq 22,\\0\leq s\leq 2}}C_{sq+i_3}
    \bigcup_{\substack{3\leq j_1 \leq 5,\\1\leq t\leq 3}}C_{tq-j_1}
    \bigcup_{\substack{12\leq j_2 \leq 14,\\1\leq t\leq 3}}C_{tq-j_2}.$$
    
    It is easy to check that $Z_1\bigcap-qZ_1=\emptyset$.\\
    
    Based on Lemma \ref{le:3.3}, we can obtain the number of entangled states $c$ in the following theorem.
    
\begin{theorem}\label{th:3.4}
    Let $n=\frac{q^2+1}{5}$ and  $ q=10k+3 $ $(k\geq 2) $  be an odd prime power. For a positive integer m with $2\leq m\leq \frac{q-3}{10}$, 
	let $\mathcal{C}$ be a cyclic code with defining set $Z$ given as follows 
	$$Z=C_0\bigcup C_1\bigcup\dots\bigcup C_{(m-1)q}.$$
	Then $|Z_{2     }|=20(m-1)^2+1$ .
\end{theorem}
\begin{proof}
	Let \\$$Z_{1}=
	\bigcup_{\substack{m\leq i_1 \leq \frac{q+2}{5}-m,\\0\leq s\leq m-2}}C_{sq+i_1}
	\bigcup_{\substack{\frac{q-3}{5}+m\leq i_2 \leq \frac{2q+4}{5}-m,\\0\leq s\leq m-2}}C_{sq+i_2}
	\bigcup_{\substack{\frac{2q-1}{5}+m\leq i_3 \leq \frac{3q+1}{5}-m,\\0\leq s\leq m-2}}C_{sq+i_3}$$
	$$\bigcup_{\substack{\frac{q-3}{5}+m\leq j_1 \leq \frac{2q+4}{5}-m,\\1\leq t\leq m-1}}C_{tq-j_1}
	\bigcup_{\substack{m-1\leq j_2 \leq \frac{q+2}{5}-m,\\1\leq t\leq m-1}}C_{tq-j_2}.
	$$
	
	and
	$$ Z_1'=\bigcup_{\substack{0 \leq i_1 \leq m-1,\\0\leq s\leq m-2}}C_{sq+i_1}
	\bigcup_{\substack{\frac{q+2}{5}-m+1\leq i_2 \leq \frac{q-3}{5}+m-1,\\0\leq s\leq m-2}}C_{sq+i_2}
	\bigcup_{\substack{\frac{2q+4}{5}-m+1\leq i_3 \leq \frac{2q-1}{5}+m-1,\\0\leq s\leq m-2}}C_{sq+i_3}$$
	$$\bigcup_{\substack{0\leq j_1 \leq m-2,\\1\leq t\leq m-1}}C_{tq-j_1}
	\bigcup_{\substack{\frac{q+2}{5}-m+1\leq j_2 \leq \frac{q-3}{5}+m-1,\\1\leq t\leq m-1}}C_{tq-j_2}
	\bigcup_{\substack{\frac{2q+4}{5}-m+1\leq j_3 \leq \frac{2q-1}{5}+m-1,\\1\leq t\leq m-1}}C_{tq-j_3}.
	\\$$
	From Lemma \ref{le:3.2}, we have 
	$$-qZ_1'=\bigcup_{\substack{0 \leq i_1 \leq m-1,\\0\leq s\leq m-2}}C_{i_1q-s}
	\bigcup_{\substack{\frac{q+2}{5}-m+1\leq i_2 \leq \frac{q-3}{5}+m-1,\\0\leq s\leq m-2}}C_{i_2q-s}
	\bigcup_{\substack{\frac{2q+4}{5}-m+1\leq i_3 \leq \frac{2q-1}{5}+m-1,\\0\leq s\leq m-2}}C_{i_3q-s}$$
	$$\bigcup_{\substack{0\leq j_1 \leq m-2,\\1\leq t\leq m-1}}C_{j_1q+t}
	\bigcup_{\substack{\frac{q+2}{5}-m+1\leq j_2 \leq \frac{q-3}{5}+m-1,\\1\leq t\leq m-1}}C_{j_2q+t}
	\bigcup_{\substack{\frac{2q+4}{5}-m+1 \leq j_3 \leq \frac{2q-1}{5}+m-1,\\1\leq t\leq m-1}}C_{j_3q+t}
	\\.$$
	
	It is easy to check that  $-qZ_1'=Z_1'$. From the definitions of $Z$, $Z_1$ and $Z_1'$, we have $Z=Z_1\bigcup Z_1'$. 
	Then from the definition of $Z_{2}$,
\begin{equation*}
\begin{split}
	Z_{2}=Z\bigcap (-qZ)&=(Z_1\bigcup Z_1')\bigcap(-qZ_1\bigcup -qZ_1')\\
	&=(Z_1\bigcap-qZ_1)\bigcup(Z_1\bigcap-qZ_1')\bigcup(Z_1'\bigcap-qZ_1)\bigcup(Z_1'\bigcap-qZ_1')\\
	&=Z_1'.
\end{split}
\end{equation*}
Therefore, $|Z_{2}|=|Z_1'|=20(m-1)^2+1$.
	
\end{proof}
    \ 
    
    From Lemmas \ref{le:3.2}, \ref{le:3.3} and Theorem \ref{th:3.4} above, we can obtain the first construction of EAQMDS codes in the following theorem.
    
\begin{theorem}\label{th:3.5}
	Let $n=\frac{q^2+1}{5}$ and  $ q=10k+3 $ $(k\geq 2) $  is an odd prime power. There are EAQMDS codes with parameters 
    $$[[n, n-4(m-1)(5m-q-5)-1, 2(m-1)q+2; 20(m-1)^2+1]]_q,$$ where $2\leq m\leq \frac{q-3}{10}$.
	
\end{theorem}
\begin{proof}
	For a positive integer $2\leq m\leq \frac{q-3}{10}$.
	Suppose that $\mathcal{C}$ is a cyclic code of length $n=\frac{q^2+1}{5}$ with defining set 
	$$Z=C_0\bigcup C_1\bigcup\dots\bigcup C_{(m-1)q}.$$
	
	Note that cyclic code $\mathcal{C}$ have $2(m-1)q+1$ consecutive roots.
	By Proposition \ref{pro:2.1}, the minimum distance of $\mathcal{C}$ is at least  $2(m-1)q+2$. 
	Then $\mathcal{C}$ is an MDS code with parameter $[[n,n-2(m-1)q-1, 2(m-1)q+2]]_{q^2}$.
	From Theorem \ref{th:3.4}, we have $|Z_{2}|=20(m-1)^2+1$.
    By Proposition \ref{pro:2.4} and Lemma \ref{le:2.6}, there are EAQEC codes with parameters
    $$[[n,n-4(m-1)(5m-q-5)-1, 2(m-1)q+2;20(m-1)^2+1]]_q.$$
    It is easy to check that $$n-k+c+2=4(m-1)q+4=2d.$$
    By Proposition \ref{pro:2.3}, which implies that the EAQEC codes are EAQMDS codes.
    
\end{proof}
    \ 

\noindent\textbf{Example 3.3}
    ~Let $q=23$ and $2\leq m\leq \frac{q-3}{10}=2.$
    Then, $n=\frac{q^2+1}{5}=106,~m=2.$ 
    According to Theorem \ref{th:3.5}, we can obtain an EAQMDS code 
$[[106,33,48;21]]_{23}.\\$

\noindent\textbf{Example 3.4}
    ~Let $q=43$ and $2\leq m\leq \frac{q-3}{10}=4.$
    Then, $n=\frac{q^2+1}{5}=370,~m=2,~3,~4$ respectively.
    According to Theorem \ref{th:3.5}, we can obtain EAQMDS codes below: 
$$[[370,217,88;21]]_{43},~~[[370,105,174;81]]_{43},~~[[370,33,260;181]]_{43}.$$\\
\noindent\textbf{Case \uppercase\expandafter{\romannumeral 2} $~~~~q=10k+7$}\\

    As for the case that $n=\frac{q^2+1}{5}$ and  $ q=10k+7 $ $(k\geq 2) $ is an odd prime power, we can produce the following EAQMDS codes. The proof is similar to that in the Case \uppercase\expandafter{\romannumeral 1}, so we omit it here.
    
\begin{lemma}\label{le:3.6}

	Let $n=\frac{q^2+1}{5}$ and  $ q=10k+7 $ $(k\geq 2) $  be an odd prime power. For a positive integer $2\leq m\leq \frac{q-7}{10}$, let 
    $$Z_{1}=\bigcup_{\substack{m\leq i_1 \leq \frac{q+3}{5}-m,\\0\leq s\leq m-2}}C_{sq+i_1}
    \bigcup_{\substack{\frac{q-2}{5}+m\leq i_2 \leq \frac{2q+1}{5}-m,\\0\leq s\leq m-2}}C_{sq+i_2}
    \bigcup_{\substack{\frac{2q-4}{5}+m\leq i_3 \leq \frac{3q+4}{5}-m,\\0\leq s\leq m-2}}C_{sq+i_3}$$
  
     $$\bigcup_{\substack{\frac{q-2}{5}+m\leq j_1 \leq \frac{2q+1}{5}-m,\\1\leq t\leq m-1}}C_{tq-j_1}
     \bigcup_{\substack{m-1\leq j_2 \leq \frac{q+3}{5}-m,\\1\leq t\leq m-1}}C_{tq-j_2}.
     $$
	Then $Z_1\bigcap-qZ_1=\emptyset$.
	
\end{lemma}

\begin{theorem}\label{th:3.7}
	Let $n=\frac{q^2+1}{5}$ and  $ q=10k+7 $ $(k\geq 2) $  be an odd prime power. For a positive integer m with $2\leq m\leq \frac{q-7}{10}$, 
	let $\mathcal{C}$ be a cyclic code with defining set $Z$ given as follows 
	$$Z=C_0\bigcup C_1\bigcup\dots\bigcup C_{(m-1)q}.$$
	Then $|Z_{2}|=20(m-1)^2+1$ .
\end{theorem}

\begin{theorem}\label{th:3.8}
	Let $n=\frac{q^2+1}{5}$ and  $ q=10k+7 $ $(k\geq 2) $  be an odd prime power. There are EAQMDS codes with parameters 
	$$[[n,n-4(m-1)(5m-q-5)-1, 2(m-1)q+2;20(m-1)^2+1]]_q,$$ where $2\leq m\leq \frac{q-7}{10}$.\\
	
\end{theorem}
\noindent\textbf{Example 3.5}
~Let $q=37$ and $2\leq m\leq \frac{q-7}{10}=3.$
Then, $n=\frac{q^2+1}{5}=274,~m=2,~3$ respectively.
\\1) Let $m=2,$ according to Lemma \ref{le:3.6}, we can obtain
\begin{equation*}
\begin{split}
Z_{1}&=\bigcup_{\substack{2\leq i_1 \leq 6,\\s=0}}C_{sq+i_1}
\bigcup_{\substack{9\leq i_2 \leq 13,\\s=0}}C_{sq+i_2}
\bigcup_{\substack{16\leq i_3 \leq 21,\\s=0}}C_{sq+i_3}
\bigcup_{\substack{9\leq j_1 \leq 13,\\t=1}}C_{tq-j_1}
\bigcup_{\substack{1\leq j_2 \leq 6,\\t=1}}C_{tq-j_2}\\
&=C_2\bigcup C_3\bigcup\dots\bigcup C_6\bigcup C_{9}\bigcup\dots\bigcup C_{13}\bigcup C_{16}\bigcup\dots\bigcup C_{21}\bigcup C_{24}\bigcup\dots\bigcup\\
&~~~~C_{28}\bigcup C_{31}\bigcup\dots\bigcup C_{36}.
\end{split}
\end{equation*}

It is easy to check that $Z_1\bigcap-qZ_1=\emptyset$.
Then, according to Theorem \ref{th:3.7}, we can obtain an EAQMDS code $[[274,401,76;21]]_{37}.$\\
2) Let $m=3,$ according to Lemma \ref{le:3.6}, we can obtain
\\$$Z_{1}=
\bigcup_{\substack{3\leq i_1 \leq 5,\\0\leq s\leq 1}}C_{sq+i_1}
\bigcup_{\substack{10\leq i_2 \leq 12,\\0\leq s\leq 1}}C_{sq+i_2}
\bigcup_{\substack{17\leq i_3 \leq 20,\\0\leq s\leq 1}}C_{sq+i_3}
\bigcup_{\substack{2\leq j_1 \leq 5,\\1\leq t\leq 2}}C_{tq-j_1}
\bigcup_{\substack{10\leq j_2 \leq 13,\\1\leq t\leq 2}}C_{tq-j_2}.$$

It is easy to check that $Z_1\bigcap-qZ_1=\emptyset$. Then, according to Theorem \ref{th:3.7}, we can obtain an EAQMDS code $[[274,489,150;81]]_{37}.$\\\\
\noindent\textbf{Example 3.6}
~Let $q=47$ and $2\leq m\leq \frac{q-7}{10}=4.$
Then, $n=\frac{q^2+1}{5}=442,~m=2,~3,~4 $ respectively.
\\1) Let $m=2,$ according to Lemma \ref{le:3.6}, we can obtain

\begin{equation*}
	\begin{split}
		Z_{1}&=
		\bigcup_{\substack{2\leq i_1 \leq 8,\\s=0}}C_{sq+i_1}
		\bigcup_{\substack{11\leq i_2 \leq 17,\\s=0}}C_{sq+i_2}
		\bigcup_{\substack{20\leq i_3 \leq 27,\\s=0}}C_{sq+i_3}
		\bigcup_{\substack{11\leq j_1 \leq 17,\\t=1}}C_{tq-j_1}
		\bigcup_{\substack{1\leq j_2 \leq 8,\\t=1}}C_{tq-j_2}\\
		&=C_2\bigcup C_3\bigcup\dots\bigcup C_8\bigcup C_{11}\bigcup\dots\bigcup C_{17}\bigcup C_{20}\bigcup\dots\bigcup C_{27}\bigcup C_{30}\bigcup\dots\\
		&~~~\bigcup C_{36}\bigcup C_{39}\bigcup\dots\bigcup C_{46}.
     \end{split}
\end{equation*}

It is easy to check that $Z_1\bigcap-qZ_1=\emptyset$.
Then, according to Theorem \ref{th:3.7}, we can obtain an EAQMDS code $[[442,609,96;21]]_{47}.$\\
2) Let $m=3,$ according to Lemma \ref{le:3.6}, we can obtain
\\$$Z_{1}=
\bigcup_{\substack{3\leq i_1 \leq 7 ,\\0\leq s\leq 1}}C_{sq+i_1}
\bigcup_{\substack{12\leq i_2 \leq 16 ,\\0\leq s\leq 1}}C_{sq+i_2}
\bigcup_{\substack{21\leq i_3 \leq 26,\\0\leq s\leq 1}}C_{sq+i_3}
\bigcup_{\substack{2\leq j_1 \leq 7,\\1\leq t\leq 2}}C_{tq-j_1}
\bigcup_{\substack{12\leq j_2 \leq 16,\\1\leq t\leq 2}}C_{tq-j_2}.$$

It is easy to check that $Z_1\bigcap-qZ_1=\emptyset$.
Then, according to Theorem \ref{th:3.7}, we can obtain an EAQMDS code $[[442,737,190;81]]_{47}.$\\
3) Let $m=4,$ according to Lemma \ref{le:3.6}, we can obtain
\\$$Z_{1}=
\bigcup_{\substack{4\leq i_1 \leq 6,\\0\leq s\leq 2}}C_{sq+i_1}
\bigcup_{\substack{13\leq i_2 \leq 15,\\0\leq s\leq 2}}C_{sq+i_2}
\bigcup_{\substack{22\leq i_3 \leq 25,\\0\leq s\leq 2}}C_{sq+i_3}
\bigcup_{\substack{3\leq j_1 \leq 6,\\1\leq t\leq 3}}C_{tq-j_1}
\bigcup_{\substack{13\leq j_2 \leq 15,\\1\leq t\leq 3}}C_{tq-j_2}.$$

It is easy to check that $Z_1\bigcap-qZ_1=\emptyset$.
Then, according to Theorem \ref{th:3.7}, we can obtain an EAQMDS code $[[442,825,284;181]]_{47}.$\\

Finally, we also have similar results for  $e>1$ is an odd positive integer. These results are given in the following lemma and theorems. Because the proofs of them are similar to that in Lemma \ref{le:3.3} and Theorems \ref{th:3.4}, \ref{th:3.5}, so we omit it here.

\begin{lemma}\label{le:3.9}
	
	Let $n=\frac{q^2+1}{5}$ and  $ q=2^e  $  be an even prime power, where $e>1$ is an odd positive integer. \\
	1) When $e\equiv 1~ {\rm mod~ }4 $ is a positive integer. For a positive integer $2\leq m\leq \frac{q-2}{10}$, let 
	$$Z_{1}=\bigcup_{\substack{m\leq i_1 \leq \frac{q+3}{5}-m,\\0\leq s\leq m-2}}C_{sq+i_1}
	\bigcup_{\substack{\frac{q-2}{5}+m\leq i_2 \leq \frac{2q+1}{5}-m,\\0\leq s\leq m-2}}C_{sq+i_2}
	\bigcup_{\substack{\frac{2q-4}{5}+m\leq i_3 \leq \frac{3q+4}{5}-m,\\0\leq s\leq m-2}}C_{sq+i_3}$$
	
	$$\bigcup_{\substack{\frac{q-2}{5}+m\leq j_1 \leq \frac{2q+1}{5}-m,\\1\leq t\leq m-1}}C_{tq-j_1}
	\bigcup_{\substack{m-1\leq j_2 \leq \frac{q+3}{5}-m,\\1\leq t\leq m-1}}C_{tq-j_2}.
	$$
	Then $Z_1\bigcap-qZ_1=\emptyset$.\\
    2) When $e\equiv 3~{\rm mod~ }4$ is a positive integer. For a positive integer $2\leq m\leq \frac{q-8}{10}$, let 
    $$Z_{1}=
	\bigcup_{\substack{m\leq i_1 \leq \frac{q+2}{5}-m,\\0\leq s\leq m-2}}C_{sq+i_1}
	\bigcup_{\substack{\frac{q-3}{5}+m\leq i_2 \leq \frac{2q+4}{5}-m,\\0\leq s\leq m-2}}C_{sq+i_2}
	\bigcup_{\substack{\frac{2q-1}{5}+m\leq i_3 \leq \frac{3q+1}{5}-m,\\0\leq s\leq m-2}}C_{sq+i_3}$$
	
	$$\bigcup_{\substack{\frac{q-3}{5}+m\leq j_1 \leq \frac{2q+4}{5}-m,\\1\leq t\leq m-1}}C_{tq-j_1}
	\bigcup_{\substack{m-1\leq j_2 \leq \frac{q+2}{5}-m,\\1\leq t\leq m-1}}C_{tq-j_2}.
	$$
	Then $Z_1\bigcap-qZ_1=\emptyset$.
	
\end{lemma}

\begin{theorem}\label{th:3.10}
	Let $n=\frac{q^2+1}{5}$ and  $ q=2^e  $  be an even prime power, where $e>1$ is an odd positive integer. \\
	1) When $e\equiv 1~ {\rm mod~ }4 $ is a positive integer. For a positive integer $2\leq m\leq \frac{q-2}{10}$,
	let $\mathcal{C}$ be a cyclic code with defining set $Z$ given as follows 
	$$Z=C_0\bigcup C_1\bigcup\dots\bigcup C_{(m-1)q}.$$
	Then $|Z_{2}|=20(m-1)^2+1$ .\\
	2) When $e\equiv 3~ {\rm mod~ }4 $ is a positive integer. For a positive integer $2\leq m\leq \frac{q-8}{10}$,
	let $\mathcal{C}$ be a cyclic code with defining set $Z$ given as follows 
	$$Z=C_0\bigcup C_1\bigcup\dots\bigcup C_{(m-1)q}.$$
	Then $|Z_{2}|=20(m-1)^2+1$ .
\end{theorem}

\begin{theorem}\label{th:3.11}
	
	Let $n=\frac{q^2+1}{5}$ and  $ q=2^e  $  be an even prime power.\\
	1) When $e\equiv 1~{\rm mod~ 4} $ is a positive integer, for a positive integer $2\leq m\leq \frac{q-2}{10}$.
	There are EAQMDS codes with parameters 
	$$[[n,n-4(m-1)(5m-q-5)-1, 2(m-1)q+2;20(m-1)^2+1]]_q.$$
	2) When $e\equiv 3~{\rm mod~ }4 $ is a positive integer, for a positive integer $2\leq m\leq \frac{q-8}{10}$.
	There are EAQMDS codes with parameters 
	$$[[n,n-4(m-1)(5m-q-5)-1, 2(m-1)q+2;20(m-1)^2+1]]_q.$$
\end{theorem}

\noindent\textbf{Example 3.7}
~Let $e=5,~q=2^5=32$ and $2\leq m\leq \frac{q-2}{10}=3.$
Then, $n=\frac{q^2+1}{5}=205,~m=2,~3$ respectively.
\\1) Let $m=2,$ according to Lemma \ref{le:3.9}, we can obtain
\begin{equation*}
\begin{split}
Z_{1}&=\bigcup_{\substack{2\leq i_1 \leq 5,\\s=0}}C_{sq+i_1}\bigcup_{\substack{8\leq i_2 \leq 11,\\s=0}}C_{sq+i_2}\bigcup_{\substack{14\leq i_3 \leq 18,\\s=0}}C_{sq+i_3}\bigcup_{\substack{8\leq j_1 \leq 11,\\t=1}}C_{tq-j_1}\bigcup_{\substack{1\leq j_2 \leq 5,\\t=1}}C_{tq-j_2}\\
&=C_2\bigcup C_3\bigcup\dots\bigcup C_5\bigcup C_{8}\bigcup\dots\bigcup C_{11}\bigcup C_{14}\bigcup\dots\bigcup C_{18}\bigcup C_{21}\bigcup\dots\\
&~~~\bigcup C_{24}\bigcup C_{27}\bigcup\dots\bigcup C_{31}.
\end{split}
\end{equation*}

It is easy to check that $Z_1\bigcap-qZ_1=\emptyset$.
Then, according to Theorem \ref{th:3.11}, we can obtain an EAQMDS code $[[205,312,66;21]]_{32}.$\\
2) Let $m=3,$ according to Lemma \ref{le:3.9}, we can obtain
\\$$Z_{1}=
\bigcup_{\substack{3\leq i_1 \leq 4,\\0\leq s\leq 1}}C_{sq+i_1}
\bigcup_{\substack{9\leq i_2 \leq 10,\\0\leq s\leq 1}}C_{sq+i_2}
\bigcup_{\substack{15\leq i_3 \leq 17,\\0\leq s\leq 1}}C_{sq+i_3}
\bigcup_{\substack{2\leq j_1 \leq 4  ,\\1\leq t\leq 2}}C_{tq-j_1}
\bigcup_{\substack{7\leq j_2 \leq 10,\\1\leq t\leq 2}}C_{tq-j_2}.$$

It is easy to check that $Z_1\bigcap-qZ_1=\emptyset$. Then, according to Theorem \ref{th:3.11}, we can obtain an EAQMDS code $[[205,380,130;81]]_{32}.$\\

\noindent\textbf{Example 3.8}
~Let $e=5,~q=2^7=128,$ and $2\leq m\leq \frac{q-8}{10}=12.$
Then, $n=\frac{q^2+1}{5}=3277,~m=2,3,\dots,12$ respectively.
\\1) Let $m=2,$ according to Lemma \ref{le:3.9}, we can obtain
\begin{equation*}
\begin{split}
Z_{1}&=\bigcup_{\substack{2\leq i_1 \leq 24,\\s=0}}C_{sq+i_1}
\bigcup_{\substack{27\leq i_2 \leq 50,\\s=0}}C_{sq+i_2}
\bigcup_{\substack{53\leq i_3 \leq 75,\\s=0}}C_{sq+i_3}
\bigcup_{\substack{27\leq j_1 \leq 50,\\t=1}}C_{tq-j_1}
\bigcup_{\substack{1\leq j_2 \leq 24,\\t=1}}C_{tq-j_2}\\
&=C_2\bigcup C_3\bigcup\dots\bigcup C_{24}\bigcup C_{27}\bigcup\dots\bigcup C_{50}\bigcup C_{53}\bigcup\dots\bigcup C_{75}\bigcup C_{78}\bigcup\\
&\dots\bigcup C_{101}\bigcup C_{104}\bigcup\dots\bigcup C_{127}.
\end{split}
\end{equation*}

It is easy to check that $Z_1\bigcap-qZ_1=\emptyset$.
Then, according to Theorem \ref{th:3.11}, we can obtain an EAQMDS code $[[3277,3768,258;21]]_{128}.$\\
2) Let $m=3,$ according to Lemma \ref{le:3.9}, we can obtain
\\$$Z_{1}=
\bigcup_{\substack{3\leq i_1 \leq 23,\\s=0}}C_{sq+i_1}
\bigcup_{\substack{28\leq i_2 \leq 49,\\s=0}}C_{sq+i_2}
\bigcup_{\substack{54\leq i_3 \leq 74,\\s=0}}C_{sq+i_3}
\bigcup_{\substack{28\leq j_1 \leq 51,\\t=1}}C_{tq-j_1}
\bigcup_{\substack{2\leq j_2 \leq 23,\\t=1}}C_{tq-j_2}$$

It is easy to check that $Z_1\bigcap-qZ_1=\emptyset$.
Then, according to Theorem \ref{th:3.11}, we can obtain an EAQMDS code $[[3277,4220,514;81]]_{128}.$\\
3) Let $m=4,5,\dots,12$ respectively, according to Theorem \ref{th:3.11}, we can obtain EAQMDS codes below:
$$[[3277,4632,770;181]]_{128},
	~~~~[[3277,5004,1026;321]]_{128},
	~~[[3277,5336,1282;501]]_{128},$$
$[[3277,5628,1538;721]]_{128},
~~[[3277,5880,1794;981]]_{128},
~~[[3277,6092,2050;1281]]_{128},$
$$[[3277,6264,2306;1621]]_{128},
[[3277,6396,2562;2001]]_{128},
[[3277,6488,2818;2421]]_{128}.$$
\section{Conclusion}
In this paper, series of EAQMDS codes with length $n=\frac{q^2+1}5$ are constructed from cyclic codes by decomposing the defining set of cyclic codes. These EAQMDS codes have much bigger minimum distance than the known quantum MDS codes and EAQECCs with the same length. It would be interesting to construct more entanglement-assisted quantum MDS codes from other classical codes with different lengths.

\end{document}